\documentclass[12pt, a4paper]{amsart} 
\usepackage[margin=3cm]{geometry}
\usepackage{ amssymb, amsthm, appendix,enumerate, float, graphicx, multirow, subcaption, tabu}
\usepackage[round,authoryear, sort]{natbib}

\usepackage{algorithm,algpseudocode}

\newtheorem{theorem}{Theorem}[section]
\newtheorem{proposition}[theorem]{Proposition}
\newtheorem{lemma}[theorem]{Lemma}
\newtheorem{corollary}[theorem]{Corollary}

\newtheorem*{theorem*}{theorem}

\theoremstyle{definition}
\newtheorem{definition}[theorem]{Definition}
\newtheorem*{definition*}{Definition}

\newtheorem{observation}[theorem]{Observation}

\newtheorem*{question*}{Question}

\newtheorem{example}[theorem]{Example} 

\DeclareRobustCommand{\stirling}{\genfrac\{\}{0pt}{}}

\newcommand{\ra}{\rightarrow}

\newcommand{\HP}{{H \negthinspace P}}
\newcommand{\MHP}{{M \negthinspace H \negthinspace P}}
\newcommand{\MRP}{{M \negthinspace R \negthinspace P}}
\newcommand{\RP}{{R \negthinspace P}}
\newcommand{\BRP}{{B \negthinspace R \negthinspace P}}

\newcommand{\HH}{\mathcal H}

\usepackage{color}

\title[a metric on phylogenetic trees]{A partial order and cluster-similarity metric on rooted phylogenetic trees}

\author{Michael Hendriksen}
\author{Andrew Francis}
\address{Centre for Research in Mathematics and Data Science, Western Sydney University, NSW, Australia}
\email{m.hendriksen@westernsydney.edu.au, a.francis@westernsydney.edu.au} 

\begin{document}

\begin{abstract}
Metrics on rooted phylogenetic trees are integral to a number of areas of phylogenetic analysis. Cluster-similarity metrics have recently been introduced in order to limit skew in the distribution of distances, and to ensure that trees in the neighbourhood of each other have similar hierarchies. In the present paper we introduce a new cluster-similarity metric on rooted phylogenetic tree space that has an associated local operation, allowing for easy calculation of neighbourhoods, a trait that is desirable for MCMC calculations. The metric is defined by the distance on the Hasse diagram induced by a partial order on the set of rooted phylogenetic trees, itself based on the notion of a hierarchy-preserving map between trees.  The partial order we introduce is a refinement of the well-known refinement order on hierarchies.  Both the partial order and the hierarchy-preserving maps may also be of independent interest. 
\end{abstract}

\date{\today}

\maketitle

\section{Introduction}
Phylogenetic trees arise frequently in attempts to describe relations among species, and it is often necessary to be able to compare trees that represent different possible relations among the same set of taxa. For instance, assigning a distance between phylogenetic trees can be important for assessing the consistency among the tree topologies inferred from different sampling of alleles (see~\cite{Zhang2019} for a recent example relating to the bamboo genus \textit{Phyllostachys}). 

Metrics are also used in a number of other areas in phylogenetics to measure dissimilarity between phylogenetic trees, such as the exploration of tree space, computation of consensus methods, and assessments of phylogenetic reconstruction. Although the earliest metric on rooted phylogenetic trees applicable to both binary and non-binary trees was defined in 1981 --- the Robinson-Foulds metric~\citep{Robinson1981} --- since the 1990's there has been a relative explosion of metrics on rooted trees, including split nodal distance~\citep{Cardona2009}, transposition distance~\citep{Alberich2009}, matching cluster distance~\citep{Bogdanowicz2013}, and a parsimony-based metric~\citep{Moulton2015}, as well as the rNNI and rSPR distances, first studied on rooted trees by \citet{Moore1973} and \citet{Hein1990} respectively (with the former only considering binary trees).

A major downside of several easily computable metrics, including the Robinson-Foulds distance, is that the majority of distances between a random pair of binary trees are comparatively very large. That is, most binary trees are as far away from each other as possible, leading to a right skew in the distribution of distances between pairs of trees in tree space~\citep{Steel1988}. This is undesirable, as it translates to a limited ability to meaningfully distinguish between binary trees.

Despite this, the Robinson-Foulds metric is well-represented in studies where a metric is required to distinguish between two or more groups of trees (recent examples include~\citet{Cole2019,Sevillya2019,Zhang2019}). This is likely due to both its ease of calculation, as well as the fact that it outperforms many other metrics on real data with respect to several measures based on practical considerations \citep{Kuhner2014}.
Additionally, metrics based on local operations such as Subtree Prune and Regraft (SPR) and Nearest Neighbour Interchange (NNI) --- often used due to the ease of calculating the neighbourhood of a given tree --- have the potentially undesirable property that trees in the neighbourhood of one another can have very different hierarchies.

In response, some new metrics based on cluster similarity have been introduced~\citep{Bogdanowicz2013,Shuguang2015} that have been shown to have fewer of the aforementioned downsides of other metrics.  In the present paper, we introduce an alternative metric based on cluster similarity, with several potential benefits. The metric is based on a graded partial order, which means the associated theory can be brought to bear and the rank can be used to estimate tree distances. It also relies on a natural local operation to move around in tree space, allowing for easy computation of the neighbourhood of a given tree --- a particularly useful property in MCMC exploration of tree space. Finally, the trees have correspondingly much larger neighbourhoods than other local operation metrics, also useful for MCMC exploration~\citep{Guo2008}. Given the widespread use of MCMC to infer phylogenies, for instance by \cite{Okumura2012}, these aspects are especially important to consider.

While calculating distances within the metric is non-trivial (the authors have not yet found a sub-exponential algorithm to do so), we provide an upper bound approximation that matches the true distance in the majority of cases in experimental simulations.  This approximation takes polynomial time and simulations suggest that the upper bound for the metric does not have a skew (unlike the Robinson-Foulds distance), so it is hoped that this metric will also not be skewed. 

As with previous cluster-similarity metrics, trees that are a short distance apart have similar hierarchies. Indeed, for any pair of trees of distance $1$ apart, the symmetric difference of their hierarchies contains at most three clusters. The metric is based upon the concept of a hierarchy-preserving map, which, as the name suggests, relates trees that have similar hierarchies. The partial order and the hierarchy-preserving maps may also be of independent interest.

Specifically, we anticipate that this new metric will outperform Robinson-Foulds metrics in discrimination between sets of trees, especially on real data as computational experiments have shown the present metric to remain successful at discrimination in the specific case of bifurcating trees. Additionally, it should increase accuracy of phylogenetic reconstruction using Markov Chain Monte Carlo methods. Finally, as the upper bound approximation is easy to calculate and relatively accurate, it will ameliorate computation speed concerns as well. 

In Section 2 we introduce the notion of a hierarchy-preserving map between trees, and show that there is a unique maximal hierarchy-preserving map between any pair of trees for which a hierarchy-preserving map exists. We then show that hierarchy-preserving maps induce a partial order on the set of rooted phylogenetic trees, and make some initial observations about the partial order, including that it refines refinement. In Section 3 we introduce a metric based on the Hasse diagram of the partial order induced by hierarchy-preserving maps. In Section 4 we introduce an algorithm for calculating an upper bound on the metric, and present initial results on its properties. Finally, in Section 5 we present some computational findings from a program used to calculate the upper bound on the metric, with the program available at \cite{Hendriksen2019}.

\section{Hierarchy-preserving maps}
\label{s:hierarchy.preserving.maps}

Throughout this paper we refer to \textit{phylogenetic trees} on a set of taxa $X$, which are rooted trees with no vertices of degree-$2$ other than the root, and whose leaves are bijectively labelled by the set $X$. The set of all such trees on a given set $X$ is denoted $RP(X)$. If all non-leaf and non-root vertices of a tree $T$ have degree $3$, $T$ is referred to as \textit{binary}, and the set of all binary trees on $X$ is denoted $BRP(X)$;.

In this section we introduce \textit{hierarchy-preserving} maps on the set of trees $RP(X)$. These are used to define a partial order on $RP(X)$.

Recall the following standard definitions in phylogenetics (see for example the book by ~\cite{PhyloSteel}).

\begin{definition}
A \textit{hierarchy} $H$ on a set $X$ is a collection of subsets of $X$ with the following properties:

\begin{enumerate}
\item $H$ contains both $X$ and all singleton sets $\{x\}$ for $x \in X$.
\item If $H_1,H_2\in H$, then $H_1 \cap H_2 = \varnothing$, $H_1 \subseteq H_2$ or $H_2 \subseteq H_1$.
\end{enumerate} 
\end{definition}

\begin{definition}
Let $T \in RP(X)$ be a tree and $v$ be a vertex of $T$. Then the \textit{cluster} of $T$ associated with $v$ is the subset of $X$ consisting of the descendants of $v$ in $T$. If a cluster $C$ is not $X$ or a singleton, $C$ is referred to as a \textit{proper cluster}, and the set of proper clusters of $T$ is denoted $P(T)$.
\end{definition}

A collection of subsets of $X$ is a hierarchy if and only if it is the set of clusters of some rooted phylogenetic tree $T$ taken over all vertices of $T$ (see~\cite{PhyloSteel} for instance). For this reason we refer to the set of clusters of $T$ as the \textit{hierarchy} of $T$, denoted $H(T)$. 

\begin{definition}\label{d:hpm}
Let $T,T' \in RP(X)$ with hierarchies $H(T)$ and $H(T')$. Then $\delta : H(T) \ra H(T')$ is a \textit{hierarchy-preserving} map if $\delta$ is the identity on singletons and the following properties hold for all $A,B \in H(T)$:
\begin{enumerate}
\item \textbf{Enveloping:} $A \subseteq \delta(A)$, and
\item \textbf{Subset-Preserving:} $A \subset B$ implies $\delta(A) \subset \delta(B)$.
\end{enumerate}

There are several interesting properties that follow almost immediately from the definitions. It is easy to check, for instance, that the composition of two hierarchy-preserving maps is also a hierarchy-preserving map. Furthermore, a hierarchy-preserving map will always map $X$ to $X$.

If $\delta: H(T) \ra H(T')$ is a hierarchy-preserving map and there exists no hierarchy preserving map $\varphi: H(T) \ra H(T')$ with $\varphi \ne \delta$ so that $\delta(A) \subseteq \varphi(A)$ for all $A \in H(T)$, then $\delta$ is termed \textit{maximal} (with respect to $T$ and $T'$).
\end{definition}

\begin{example}
Let $T,T' \in RP(X)$ where $X=\{a,b,c,d,e,f\}$ as depicted in Figure \ref{ExampleHPM}. Then $P(T)=\{ab,cd,abcd\}$ and $P(T') = \{abcd,abcde\}$. Then there exists a hierarchy-preserving map $\varphi$ from $H(T)$ to $H(T')$ that is the identity on singletons and $X$, maps $ab$ and $cd$ to $abcd$ and maps $abcd$ to $abcde$. One can easily confirm the necessary properties hold, and that this is the unique hierarchy-preserving map from $T$ to $T'$. 
\end{example}

\begin{figure}[ht]
\centering
\includegraphics{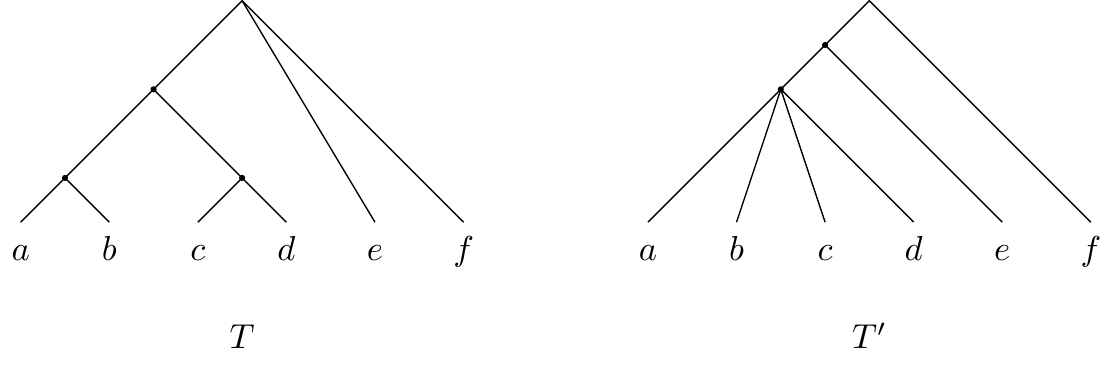}
\caption{A pair of trees $T$ and $T'$ with a hierarchy-preserving map from $H(T)$ to $H(T')$ that maps $ab$ and $cd$ to $abcd$, and maps $abcd$ to $abcde$.}
\label{ExampleHPM}
\end{figure}

\begin{theorem}
\label{t:MaximalMap}
For $T,T' \in RP(X)$, if there exists a hierarchy-preserving map from $T$ to $T'$ then there is a unique maximal hierarchy-preserving map from $T$ to $T'$.
\end{theorem}

\begin{proof}
Suppose that $\delta_1: H(T) \ra H(T')$ and $\delta_2: H(T) \ra H(T')$ are distinct maximal hierarchy preserving maps. As they are distinct, there must be a cluster $B$ of $H(T)$ such that $\delta_1$ and $\delta_2$ disagree. In particular, since at the very least $\delta_1(X)=\delta_2(X)=X$, there must be some non-singleton cluster $B$ so that $\delta_1$ and $\delta_2$ disagree, but $\delta_1$ and $\delta_2$ agree on all clusters containing $B$. Denote the inclusion-minimal cluster containing $B$ in $H(T)$ by $C$. Now, as $\delta_1,\delta_2$ are enveloping, both $\delta_1(B)$ and $\delta_2(B)$ contain $B$. Therefore either $\delta_1(B) \subset \delta_2(B)$ or vice versa. Assume without loss of generality that $\delta_1(B) \subset \delta_2(B)$. Define $\delta_1': H(T) \ra H(T')$ as follows:

\begin{equation*}
\delta_1'(M)=\begin{cases}
\delta_1(M), & \text{if $M \ne B$}\\
\delta_2(B), & \text{if $M = B$}.\\
\end{cases}
\end{equation*}

We will show that this is a hierarchy-preserving map, which contradicts the maximality of $\delta_1$. It follows that there is a unique maximal hierarchy-preserving map.

We can immediately see that $\delta_1'$ is certainly enveloping, as for $M \ne B$ we can use the fact that $\delta_1$ is enveloping, and for $M=B$ we can use that $\delta_2$ is enveloping.

We will now prove that $\delta_1'$ is subset-preserving. First suppose $M \subset B$. Then $\delta_1'(M) = \delta_1(M) \subset \delta_1(B) \subseteq \delta_2(B) = \delta_1'(B)$, by definition of $\delta_1'$ and the fact that $\delta_1$ is subset-preserving. Now, suppose that $B \subset M$. As $B$ is inclusion-maximal in $C$, this means that $M \supseteq C$, and we know that $\delta_1'(B) = \delta_2(B) \subset \delta_2(C) = \delta_1(C) \subseteq \delta_1(M)$ by definition of $\delta_1'$ and the fact that $\delta_1$ is subset-preserving again. Hence $\delta_1'$ is subset-preserving.

Thus we have found a hierarchy preserving map $\delta_1': H(T) \ra H(T')$ with $\delta_1' \ne \delta_1$ for which $\delta_1(A) \subseteq \delta_1'(A)$ for all $A \in H(T)$, contradicting the maximality of $\delta_1$. It follows that there is a unique maximal hierarchy preserving map from $T$ to $T'$.
\end{proof}

We now use the hierarchy-preserving maps just introduced, to define a partial order $\le_\HP$ on $RP(X)$. We say $T \le_\HP T'$ if there is a hierarchy-preserving map from $H(T)$ to $H(T')$. We will make use of the notion of a ``maximal vertical subhierarchy'', as defined below.

\begin{definition}
Let $T \in RP(X)$. Let $C_1$ be a cluster in $H(T)$, and suppose that $C_1,\dots,C_k$ are distinct clusters in $H(T)$ with the property that $C_1 \subset \dots \subset C_k$ and there are no other clusters $D$ for which $C_i \subset D \subset C_{i+1}$. Then we say $\{C_1,\dots,C_k\}$ is a \textit{maximal vertical subhierarchy} of $C_1$ in $H(T)$.
\end{definition}

\begin{theorem}\label{t:poset}
The set $RP(X)$ forms a poset under the relation $\le_\HP $.
\end{theorem}

\begin{proof}
The observation that the identity map from the hierarchy of a tree to itself is a hierarchy-preserving map gives reflexivity, and the transitivity of hierarchy-preserving maps is also easy to check. It remains to show antisymmetry. 

Suppose $T\le_\HP T'$ and $T'\le_\HP T$. Then there exist hierarchy-preserving maps $\varphi_1:H(T) \ra H(T')$ and $\varphi_2: H(T') \ra H(T)$. We claim that both must be the identity mapping.

Suppose, seeking a contradiction, that $\varphi_1$ is not an identity mapping. Then there must be some cluster $C_1 \in H(T)$ so that $\varphi_1(C_1) =D_1 \ne C_1$, and as $\varphi_1(X) =X$, we can choose $C_1$ such that all clusters containing $C_1$ are mapped to themselves under $\varphi_1$ - that is, $\varphi_1$ acts as the identity on all elements of the maximal vertical subhierarchy $C_1,\dots,C_k$ of $C_1$ except $C_1$ itself. In particular this implies that $C_2,\dots,C_k$ are all clusters of $T'$ as well. 

We first show that $\varphi_2$ is the identity on $C_2,\dots,C_k$. Let $C_i$ be the inclusion-maximal element in this maximal vertical subhierarchy for which $\varphi_2(C_i) \ne C_i$. As $\varphi_2$ is subset-preserving, $\varphi_2(C_i)$ must be some inclusion-maximal subcluster of $C_{i+1}$, and as $\varphi_2$ is enveloping $C_i \subseteq \varphi_2(C_i)$. But $C_1,\dots,C_k$ is a maximal vertical subhierarchy of $T$ and so this means $\varphi_2(C_i) = C_i$, a contradiction. Therefore $\varphi_2$ is the identity on $C_2,\dots,C_k$.

We now finally consider $\varphi_2(D_1)$. As $\varphi_1(C_1)=D_1$ and $\varphi_1$ is enveloping, $C_1 \subset D_1 \subset \varphi_2(D_1)$. Therefore $\varphi_2(D_1)$ must be an element of the maximal vertical subhierarchy of $C_1$, which by subset-preservation and the fact that $\varphi_2$ is the identity on $C_2,\dots,C_k$ forces $\varphi_2(D_1)=C_1$. But then we get that $C_1 \subseteq D_1 \subseteq \varphi_2(D_1) = C_1$ and hence $C_1 = D_1$, contradicting the assumption that $\varphi_1(C_1) =D_1 \ne C_1$. It follows that $\varphi_1$ is the identity mapping and so $T = T'$, giving antisymmetry, and completing the proof.
\end{proof}

For several results in the remainder of this section, we will show given two trees $T \le_\HP T'$, how to construct a tree $T''$, so that $T \le_\HP T'' \le_\HP T'$. 
The tree we construct will be a ``binding'' of $T$.

\begin{definition}
Let $T \in RP(X)$, and let $A_1, \dots, A_m\in H(T)$ (with $m \ge 2$) be distinct inclusion-maximal subclusters of a cluster $D\in H(T)$ such that $\bigcup_{i=1}^m A_i \ne D$. Take $H(T)$, delete all $A_i$ for which $|A_i| > 1$ from $H(T)$, and add $\bigcup_{i=1}^m A_i$, forming a new set of clusters,
\[
\mathcal H:=\left(H(T)\setminus\{A_i \colon |A_i|>1\}\right)\cup\left\{\bigcup_{i=1}^m A_i\right\}.
\]
Then $\mathcal H$ is a hierarchy (see Lemma \ref{UnionHP}), and the corresponding tree is termed a \textit{binding} of $T$ at $\bigcup_{i=1}^m A_i$, and denoted $T_{\bigcup_{i=1}^m A_i}^D$. If a tree $T'$ can be obtained from $T$ by binding, then $T$ is termed an \textit{unbinding} of $T'$.
\end{definition}

\begin{example}
\label{BindingExample}
Let $X=\{a,b,c,d,e,f,g,h \}$ and let $T \in RP(X)$ be such that $P(T) = \{ ab,abc,de, abcde\hspace{-0.1em}f\hspace{-0.1em}g \}$. Let $A= abcde$, $B=abcde\hspace{-0.1em}f$ and $D=abcde\hspace{-0.1em}f\hspace{-0.1em}g$. Then the binding of $T$ at $A$, denoted $T_A^D$, is the tree on $X$ corresponding to the hierarchy with proper clusters $ab,abcde,abcde\hspace{-0.1em}f\hspace{-0.1em}g$. The binding of $T$ at $B$, denoted $T_B^D$, is the tree on $X$ corresponding to the hierarchy with proper clusters $ab,abcde\hspace{-0.1em}f,abcde\hspace{-0.1em}f\hspace{-0.1em}g$; specifically, note that we do not delete $f$ as it is a singleton and the result would no longer be a hierarchy. These three trees can be seen in Figure \ref{Binding}.
\end{example}

\begin{figure}[ht]
\centering
\includegraphics{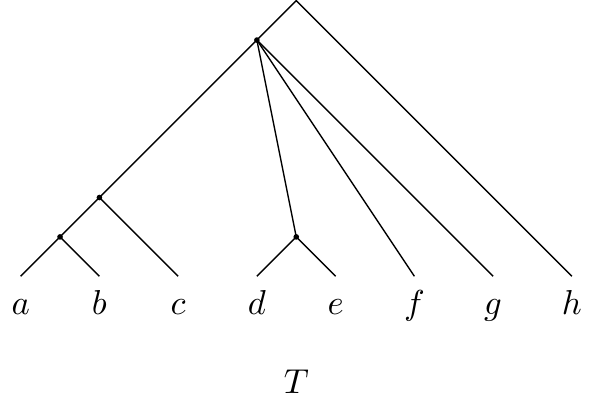}
\par\medskip
\includegraphics{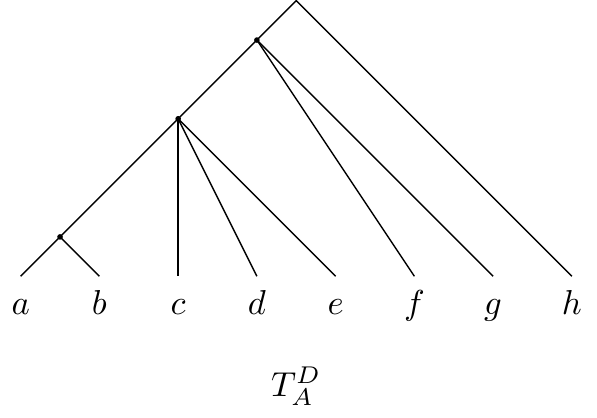}
\hspace{1cm}
\includegraphics{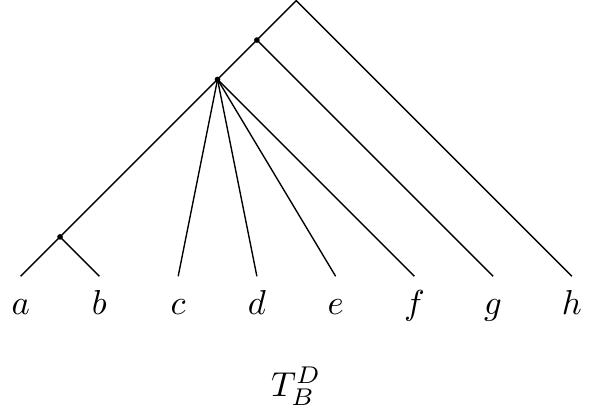}
\caption{Two potential bindings of the tree $T$, as described in Example \ref{BindingExample}, with $A=abcde$, $B=abcde\hspace{-0.1em}f$, and $D=abcde\hspace{-0.1em}f\hspace{-0.1em}g$.}
\label{Binding}
\end{figure}

\begin{lemma}
\label{UnionHP}
Let $T\in RP(X)$, and suppose $A,B$ are distinct inclusion-maximal subclusters of some cluster $D\in H(T)$, where $D \ne A \cup B$. Then the binding of $T$ at $A \cup B$ is a hierarchy. Moreover, if $T_{A \cup B}^D$ is the corresponding tree, then $T <_\HP T_{A \cup B}^D$.
\end{lemma}

\begin{proof}
In a minor abuse of notation, let $H(T_{A \cup B}^D)$ be the set of clusters corresponding to the binding of $T$ at $A \cup B$. To confirm that $H(T_{A \cup B}^D)$ is a hierarchy, it suffices to check that any $M \in H(T_{A \cup B}^D)$ for which $M \cap (A \cup B) \ne \varnothing$ is either contained in or contains $A \cup B$.

If $M \cap (A \cup B)$ is non-empty, then $M \cap A$ or $M \cap B)$ is non-empty. Hence, since $M$ is a cluster in $H(T)$, and as $A,B$ are inclusion-maximal in $D$, it follows that $M$ either contains $D$ (and so contains $A \cup B$), or is a subset of $A$ or $B$ (and thus is contained in $A \cup B$). Thus $H(T_{A \cup B}^D)$ is a hierarchy.

The second statement of the lemma follows for two reasons. Firstly, as $A \cup B$ is certainly not a cluster in $T$ we know that $T \ne T_{A \cup B}^D$. Secondly, because the map from $H(T)$ to $H(T_{A \cup B}^D)$ that is the identity on all clusters except for $A$ and $B$, which are mapped to $A \cup B$, is clearly hierarchy-preserving.
\end{proof}

\begin{theorem}
\label{t:Union2OfMore}
Suppose $T \le_\HP T'$ and $\delta: H(T) \ra H(T')$ is a maximal hierarchy preserving map. If $A,B,C$ are three inclusion-maximal subclusters of some cluster $D$, and $\delta(A)=\delta(B)$ contains $A \cup B \cup C$, then $T <_\HP T_{A \cup B}^D < T'$.
\end{theorem}

\begin{proof}
By Lemma \ref{UnionHP} we know that the set of clusters $H(T^D_{A\cup B})$ is a hierarchy, and that $T <_{\HP} T_{A\cup B}^D$. We can also see that $T_{A\cup B}^D \ne T'$ - if they were equal, the maximal hierarchy-preserving map $\delta$ from $T$ to $T_{A\cup B}^D = T'$ must be the identity on all clusters of $H(T)$ except $A$ and $B$, and map both $A$ and $B$ to $A \cup B$. But then $\delta(A)$ could not contain $A \cup B \cup C$, contradicting the assumptions of the theorem and showing $T_{A\cup B}^D \ne T'$.

It therefore suffices to show that there is a hierarchy-preserving map $\delta ' : H(T_{A\cup B}^D) \ra H(T')$. Noting that all clusters in $H(T_{A\cup B}^D)$ other than $A\cup B$ are also clusters in $H(T)$, for any cluster $M\in H(T_{A\cup B}^D)$, define
\begin{equation*}
\delta'(M)=\begin{cases}
\delta(M), & \text{if $M \ne {A\cup B}$}\\
\delta(A), & \text{if $M = {A\cup B}$}.\\
\end{cases}
\end{equation*}
We claim that $\delta'$ is a hierarchy-preserving map from $H(T_{A\cup B}^D)$ to $H(T')$ as required.

Certainly $\delta'$ is enveloping as $\delta$ is enveloping (so $M \subseteq \delta'(M)$ for all $M \ne A \cup B$), and $\delta(A\cup B) = \delta(A) \supseteq (A \cup B \cup C) \supset (A \cup B)$.

We now check subset preservation. For $Y$ and $Z$ clusters in $H(T_{A\cup B}^D)$, we need to check $Y\subset Z$ implies $\delta'(Y)\subset\delta'(Z)$. If neither are equal to $A\cup B$, then this follows immediately from the definition of $\delta'$ and the properties of $\delta$. It remains to check the two cases: (i) $Y=A\cup B\subset Z$, and (ii) $Y\subset A\cup B=Z$.

In the first case, $A\cup B \subset Z$ implies $D \subseteq Z$, because $A$ and $B$ are inclusion-maximal subclusters of $D$. Then $\delta'(A\cup B)=\delta(A)$ by definition of $\delta'$, and $\delta(A) \subset \delta(D) \subseteq \delta(Z)$ because $\delta$ is subset-preserving and $B,D,Z$ are all clusters in $H(T)$. Finally noting that $\delta'(Z)=\delta(Z)$ completes this case.

In the second case, $Y\subset A\cup B$ implies $Y\subset A$ or $Y\subset B$ because $Y, A, B$ are all part of a single hierarchy, $H(T)$. Assuming without loss of generality that $Y\subset A$, we have: $\delta'(Y)=\delta(Y)\subset\delta(A)$ by subset-preservation of $\delta$; and $\delta(A)=\delta'(A\cup B)$ by definition of $\delta'$. Therefore $\delta'(Y)\subset \delta'(A\cup B)=\delta'(Z)$, as required. 
\end{proof}

We finish this section with a result describing the maximal elements under the partial order $\le_\HP$. Note that the $\le_\HP$-minimal element is the star tree.

\begin{proposition}\label{p:max.elts.RPX}
The set of $\le_\HP$-maximal elements of $RP(X)$ is precisely $BRP(X)$, the set of binary trees.
\end{proposition}

\begin{proof}
First, if a tree is non-binary, then its hierarchy has a cluster with at least three inclusion-maximal subclusters. Therefore, by Theorem~\ref{t:Union2OfMore}, we can bind two of them to create a tree that is strictly greater in the partial order. So non-binary trees are not $\le_\HP$-maximal.

Second, if two trees $T$ and $T'$ are binary and there is a hierarchy-preserving map between them, they must be equal, as follows.

Let $\varphi:H(T) \ra H(T')$ be a hierarchy-preserving map. Observe that $\varphi$ maps $X$ to $X$ (by definition of a hierarchy-preserving map), and let $Y$ be a non-singleton cluster of $T$ such that for every cluster $Z$ in the maximal vertical subhierarchy of $Y$, $\varphi(Z)=Z$ . As $T$ and $T'$ are binary, $Y$ has two inclusion-maximal subclusters in each of $H(T)$ and $H(T')$. Let $C_1$ and $C_2$ be the inclusion-maximal clusters of $Y$ in $H(T)$, and $D_1$ and $D_2$ be the inclusion-maximal clusters of $Y$ in $H(T')$. As $\varphi$ is subset-preserving, $C_1$ and $C_2$ must each be mapped to some subcluster of $D_1$ and $D_2$. As $\varphi$ is enveloping, this implies that each of $C_1$ and $C_2$ are subsets of $D_1$ or $D_2$. Additionally, $C_1 \cup C_2 = Y = D_1\cup D_2$, which forces $C_1=D_1$ and $C_2=D_2$ or $C_1=D_2$ and $C_2=D_1$. It follows that $\varphi$ is the identity on all elements of $H(T)$, so $T=T'$.
\end{proof}

We will often consider the partial order restricted to the set of trees below every element of a set of trees $P$. 

If $P=\{ T,\dots,T_k \}$ is a set of trees, then the set of trees $T$ for which there exists a hierarchy-preserving map $\delta_i:H(T) \ra H(T_i)$ for each $i$ is denoted by $HP(P)$. In other words, 
\[
HP(P):=\{T\in RP(X)\mid T\le_\HP T_i, \text{for all }T_i\in P\}.
\]

Recall the following standard definition in phylogenetics

\begin{definition}
Let $T, T'$ be rooted phylogenetic trees on the same set $X$. Then if every cluster of $T$ is a cluster of $T'$, $T'$ is referred to as a \textit{refinement} of $T$, denoted $T \preceq T'$.
\end{definition}

In particular, observe that if $T$ is the star tree $S$ or $T'$ is a refinement of $T$, then a hierarchy-preserving map from $T$ to $T'$ will always exist, namely the identity map on clusters in $T$. Therefore $HP(P)$ is always non-empty, as it will certainly contain $S$. We further note that if $P$ consists of the single tree $T$, then $HP(P)$ can immediately be seen to be a bounded lattice, with least element $S$ and greatest element $T$, as every element of $HP(P)$ has a hierarchy-preserving map into $T$ by definition. It follows that if $P=(T,\dots,T_k)$, then $HP(P)$ forms the poset obtained by taking the intersection of the bounded lattices corresponding to each tree in $P$. 

In fact, as $T'$ being a refinement of $T$ implies there is a hierarchy-preserving map from $T$ to $T'$, the partial order $\le_\HP$ actually \textit{refines} refinement. By this we mean that if $T \preceq T'$, then $T \le_\HP T'$, or equivalently, that edges in $RP(X)$ under the refinement partial order correspond to paths in $RP(X)$ under $\le_\HP$ that consist either entirely of up-moves or entirely of down-moves.

\begin{proposition}
Let $T \preceq T'$ in $RP(X)$. Then $T \le_\HP T'$ in $RP(X)$ .
\end{proposition}

The converse of this proposition is not true, that is, the existence of a hierarchy-preserving map from $T$ to $T'$ does not imply refinement. One can see this, for example, from either binding in Figure \ref{Binding}.

\section{An induced metric on the set of rooted phylogenetic trees}

The hierarchy-preserving maps, and associated partial order on the set of phylogenetic trees, allow us to define a new metric on the set of rooted phylogenetic trees. In this section we set out the metric, and prove some of its key properties, including information about the neighbourhood of a tree and the diameter of the space.

Let $\mathcal{H}(X)$ denote the Hasse diagram of $RP(X)$ under $\le_\HP$. That is, $\mathcal{H}(X)$ is the symmetric directed graph $(RP(X),E)$ where $(T_1,T_2) \in E$ if and only if  for either $i=1,j=2$ or $i=2,j=1$, we have $T_i \le_\HP T_j$ and for any tree $T_3$ such that $T_i \le_\HP T_3 \le_\HP T_j$, either $T_3= T_i$ or $T_3=T_j$ (that is, $T_j$ covers $T_i$ under the $\le_\HP$ relation). We then define the distance $d_\HP(T,T')$ to be the geodesic distance from $T$ to $T'$ in $\mathcal{H}(X)$. We know that $\mathcal{H}(X)$ is connected as every tree has a path to the star tree, so $d_\HP$ is certainly a metric.

The following theorem shows that if two trees are distance one apart in $\HH(X)$, then one is a binding of the other - in particular the binding of a pair of clusters in the hierarchy.

\begin{theorem}
\label{OneMove}
Let $T,T'$ be trees. Then $d_\HP(T,T')=1$ iff $T'=T_{A \cup B}^V$, for some pair of distinct clusters $A,B$ that are inclusion-maximal in $V$ in $H(T)$, or the reverse. 
\end{theorem}

\begin{proof}
Suppose first that $d_\HP(T,T')=1$ and without loss of generality that $T \le_\HP T'$. Then $T'$ covers $T$ under $\le_\HP$, that is, for any tree $T''$ such that $T \le_\HP T'' \le_\HP T'$, either $T''= T$ or $T''=T'$.
Let $\delta: H(T) \ra H(T')$ be the maximal hierarchy-preserving map between them, as defined in Definition \ref{d:hpm}.

Now, let $C$ be a cluster common to $T$ and $T'$ such that the maximal vertical subhierarchy of $C$ is common to both trees, and contains $X$, but that the inclusion-maximal subclusters of $C$ are different in $T$ and $T'$. Such a cluster always exists since $C=X$ is possible. Denote the distinct inclusion-maximal subclusters of $C$ in $H(T)$ by $A_1,\dots,A_j$, and the distinct inclusion-maximal subclusters of $C$ in $H(T')$ by $B_1,\dots,B_k$. 

The hierarchy-preserving map $\delta:H(T) \to H(T')$ acts as the identity on each element of the maximal vertical subhierarchy of $C$, for the following reasons. If $\delta$ is the identity on any cluster $D$, and that $D'$ is a subcluster of $D$ in both trees, then $D'$ must map to a subcluster of $D$ (by subset-preservation), that also contains $D'$ (enveloping). This forces $D'$ in $T$ to map to $D'$ in $T'$. Since $\delta$ acts as the identity on $X$, this forces it to act as the identity on the whole maximal vertical subhierarchy. 

Considering the subclusters of $C$ in $T$ and $T'$, this means that $\delta(A_h)= B_i$ for some unique $B_i$, and thus that $A_h \subseteq B_i$. Furthermore, each $B_i$ must be the union of some subcollection of the $A_h$'s. 

Suppose there is some $B_i$ that is the union of more than two $A_h$'s. Then by Lemma \ref{UnionHP} there exists a binding of two of those $A_h$'s that produces a tree that also maps into $T'$, contradicting the fact that $d_\HP(T,T')=1$. Hence each $B_i$ is the union of at most two $A_h$'s.

As $T \ne T'$, there must exist at least one such cluster, so suppose $B_j = A_k \cup A_\ell$. Now, suppose that there is any other cluster $A \in H(T)$ such that $\delta(A) \ne A$, or any cluster $B \in H(T')$ that is not the image of some cluster in $H(T)$. Then the binding $T_{A_k \cup A_\ell}$ is certainly different from both $T$ and $T'$, but we can see that $T <_\HP T_{A_k \cup A_\ell} <_\HP T'$, which is a contradiction as $d_\HP(T,T')=1$. It follows that the only difference between the hierarchies of $T$ and $T'$ is that $T$ contains $A_k$ and $A_\ell$ while $T'$ contains $B_j$, and the result follows. 

We now suppose, without loss of generality, that $T'=T_{A \cup B}^V$, for some pair of clusters $A,B$ that are inclusion-maximal in $V$ in $H(T)$. Then certainly $T \le_\HP T'$, so in order to show $d(T,T')=1$ it only remains to show that $T'$ covers $T$ - that is, that  if there is a tree $T''$ so that $T \le_\HP T'' \le_\HP T'$, then $T''=T$ or $T'' = T'$.

Let $T''$ be a tree so that $T \le_\HP T'' \le_\HP T'$, and let $\varphi_1: H(T) \ra H(T'')$ and $\varphi_2: H(T'') \ra H(T')$ be hierarchy-preserving maps. By Theorem \ref{t:MaximalMap}, there is a unique maximal hierarchy-preserving map $\varphi_{\max} : H(T) \ra H(T')$, and this must certainly be the map that is the identity on all clusters of $T$ except for $A$ and $B$, which are mapped to $A \cup B$ in $T'$. The composition of two hierarchy-preserving maps is also a hierarchy-preserving map, so $\varphi_2 \circ \varphi_1$ is a hierarchy-preserving map too, and due to $\varphi_{\max}$ being maximal we have that $\varphi_2 \circ \varphi_1(A) \subseteq \varphi_{\max}(A)$ for all clusters $A$ in $T$. Therefore $\varphi_2 \circ \varphi_1$ is the identity on all clusters of $T$ except for $A$ and $B$, and $\varphi_2 \circ \varphi_1(A)=\varphi_2 \circ \varphi_1(B) \subseteq A \cup B$. Furthermore, this implies that 
\[H(T) \cap H(T') \cap H(T'') = H(T) \backslash \{A,B\} = H(T') \backslash \{A \cup B\}\]
and both $\varphi_1$ and $\varphi_2$ are the identity on this intersection. 

There are two possibilities - either $\varphi_1(A) \cap \varphi_1(B) = \varnothing$ or not.

\begin{enumerate}
\item \textit{$\varphi_1(A) \cap \varphi_1(B) = \varnothing$:} As $\varphi_1$ is enveloping, $B \subseteq \varphi_1(B)$ and therefore $\varphi_1(A) \cap B = \varnothing$. But $A \subseteq \varphi_1(A) \subseteq \varphi_2 \circ \varphi_1(B) \subseteq A \cup B$, so $\varphi_1(A)=A$. Similarly, $\varphi_1(B)=B$. 

Let $M$ be some cluster of $H(T'')$. If $\varphi_2(M) \ne A \cup B$, then $\varphi_2(M) = C$ for some $C$ in $H(T') \backslash \{A \cup B \} = H(T) \cap H(T') \cap H(T'')$. But $\varphi_2$ is the identity on all elements of this intersection, so $C \in H(T)$. On the other hand, if $\varphi_2(M) = A \cup B$, as $\varphi_2$ is enveloping $M \cap A$ or $M \cap B$ is non-empty.  Then as $M$ and $A$ are both clusters in the same hierarchy $H(T'')$, so $M$ contains or is contained in $A$ or $B$. But if $M$ strictly contains or is strictly contained in $A$ or $B$, then $M$ could not map to $A \cup B$ as $\varphi_2(A) = \varphi_2(B) = A \cup B$ and this would contradict subset-preservation. This leads us to conclude that $M=A$ or $M=B$, which are again in $H(T)$. Therefore every cluster in $H(T'')$ is in $H(T)$, so as $T \le_\HP T''$ this gives us $T'' = T$.

\item \textit{$\varphi_1(A) \cap \varphi_1(B) \ne \varnothing$:} Without loss of generality we can assume $\varphi_1(A) \supseteq \varphi_1(B)$, then as $\varphi_1$ is enveloping $A \cup B \subseteq \varphi(A)$. Furthermore, as $\varphi_2 \circ \varphi_1(A) \subseteq A \cup B$ this forces $\varphi_1(A) = A\cup B$. As $H(T'')$ contains every cluster of $H(T')$ and $T'' \le_\HP T'$ it follows that $T'' = T'$ .
\end{enumerate}

As $T'' = T$ or $T'' = T'$, it follows $T'$ covers $T$ under $\le_\HP$ and hence that $d(T,T')=1$.
\end{proof}

For the rest of this section we will focus on movement around the Hasse diagram of trees, $\HH(X)$.

\begin{definition}
Let $T,T'$ be trees in $RP(X)$, and $e = (T,T') \in E(\HH(X))$. Then $e$ is referred to as an \textit{up-move} if $T \le_\HP T'$ and a \textit{down-move} if $T' \le_\HP T$.
\end{definition}

Note that by Theorem \ref{OneMove}, an up-move takes one from a tree to a binding of two clusters of that tree (that are inclusion-maximal in some third cluster), and a down move does the reverse. See Figures \ref{f:upmoves} and \ref{f:downmoves} for some examples.

Let us now clearly elucidate what a down-move actually \textit{does}. One can consider the up-move to be the deletion of some distinct pair of clusters $A,B \in H(T)$ that are inclusion-maximal in a third cluster $C$, with $A\cup B\subsetneq C$ (unless $A$ or $B$ are singletons in which case only non-singletons are deleted) and then the addition of $A \cup B$. 

\begin{figure}[ht]
\centering
\subcaptionbox{Up-move without singletons}{ \includegraphics[width = 0.84\textwidth]{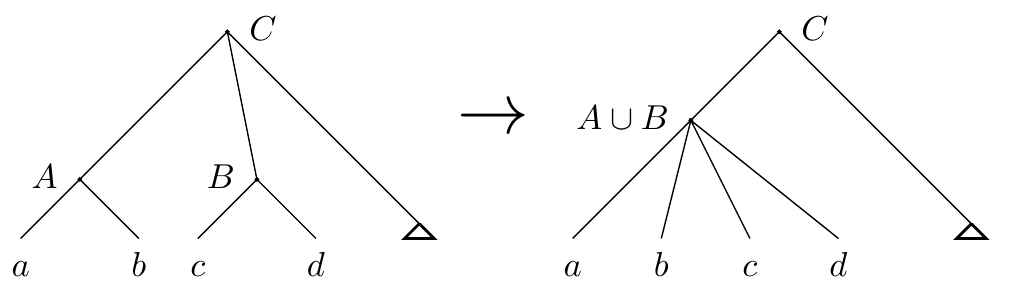}}
\par\bigskip
\subcaptionbox{Up-move with a singleton}{ \includegraphics[width = 0.84\textwidth]{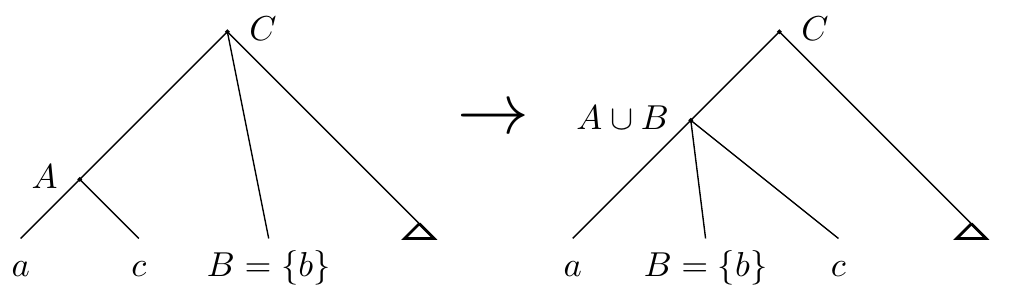}}
\caption{Examples of up-moves. The up-moves in (A) show one example without singleton clusters, and in (B) one in which one of the clusters is a singleton (it is also possible for both to be singletons). In all cases, a bold triangle indicates a non-singleton cluster.}
\label{f:upmoves}
\end{figure}

A down-move is therefore the reverse of this. To be precise, we select some cluster $Z \in H(T)$ with distinct inclusion-maximal clusters $Y_1,\dots,Y_k$. We then partition these inclusion-maximal clusters into two, to form (after relabelling) $\bigcup_{i=1}^j Y_i$ and $\bigcup_{i=j+1}^k Y_i$, under the restriction that each union can only contain one element if that element is a singleton. Then, we add the clusters from $\{\bigcup_{i=1}^j Y_i, \bigcup_{i=j+1}^k Y_i\}$ that are not singletons, and delete $Z$.

\begin{figure}[ht]
\subcaptionbox{Down-move with unions of multiple clusters}{ \includegraphics[width = 0.84\textwidth]{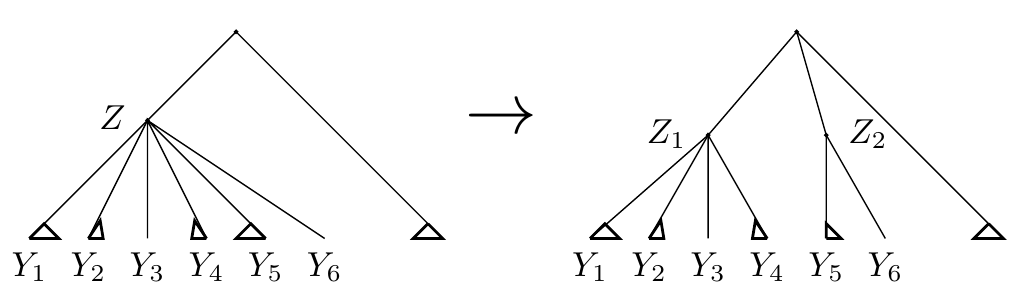}}
\par\bigskip
\subcaptionbox{Down-move with a union and a singleton}{ \includegraphics[width = 0.84\textwidth]{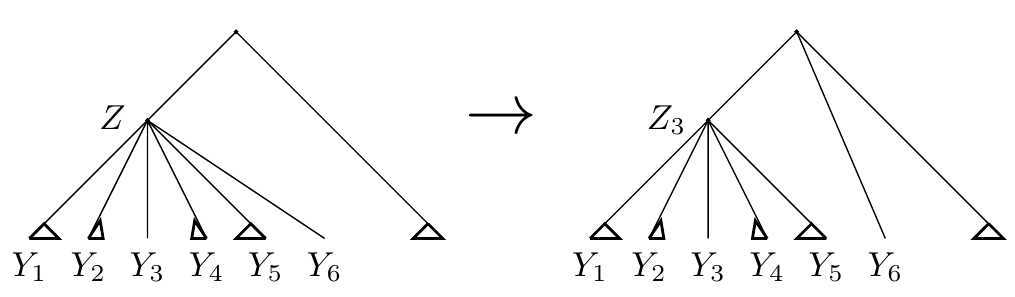}}
\caption{Examples of down-moves. The down-moves in (A) show one example in which each union contains more than one cluster, and in (B) one in which one union is just a single cluster, in which case it must be a singleton (here $Y_6$). In all cases, a bold triangle indicates a non-singleton cluster.}
\label{f:downmoves}
\end{figure}

For a tree $T$, recall that $P(T)$ is the set of proper clusters of the hierarchy corresponding to $T$, and let 

\[
f(T) = \left( \displaystyle\sum_{A \in P(T)} |A| \right) - |P(T)|= \displaystyle\sum_{A \in P(T)} \left( |A| -1 \right), 
\]
noting that this number will always be non-negative, and will only be zero if $T$ is the star tree, in which case $P(T)=\emptyset$.

We call $f(T)$ the \emph{rank} of $T$. The rank provides an easy shortcut to calculating the distance between certain trees, if one is above the other in $\HH(X)$:

\begin{theorem}
\label{t:verticaldistance}
If $T,T' \in RP(X)$, with $T \le_\HP T'$, then 
\[
d_\HP(T,T')= f(T') - f(T).
\]
\end{theorem}

\begin{proof}
Recall that an up-move corresponds to taking the union of two distinct clusters $A,B$ that are inclusion-maximal in some cluster $C$ and deleting $A$ if $|A| >1$ and deleting $B$ if $|B| >1$. 

Let $T,T' \in RP(X)$ and $\delta: H(T) \ra H(T')$ a maximal hierarchy-preserving map between them. For $A \in H(T')$, let $\delta^{-1}(A)$ denote the set of clusters that map to $A$, and let $c_A:=|\delta^{-1}(A)|$. 
We can see that for each cluster $A \in H(T')$ for which $c_A > 1$, we can bind the clusters in $\delta^{-1}(A)$ to form $\bigcup_{B \in \delta^{-1}(A)} B$, which will take $c_A -1$ moves.

As $\delta$ is maximal, all elements of $\delta^{-1}(A)$ are inclusion-maximal in some cluster $C$. We will then need to bind each singleton element of $A \setminus \bigcup_{B \in \delta^{-1}(A)} B$ with $B$, which will take $|A|-\left| \bigcup_{B \in \delta^{-1}(A)} B\right|$ moves (which will again always form a tree due to maximality of $\delta$)

It follows that it takes 

\[
\left( c_A + |A|-\left| \bigcup_{B \in \delta^{-1}(A)} B\right| - 1 \right)
\]
moves to obtain $A$ using this method. 

If $c_A = 0$, then we can form a subcluster of size $2$ of $A$, then add the remaining elements of $A$ one at a time, which will require $|A|-1$ moves. Observe that in this case $c_A=0$ and $|\displaystyle\bigcup_{B \in \delta^{-1}(A)} B|=0$. so 

\[
\left( c_A + |A|-| \displaystyle\bigcup_{B \in \delta^{-1}(A)} B| - 1 \right) = |A| -1.
\]

It follows that using this method (starting with inclusion-maximal proper clusters of $H(T)$ and working our way down, so that we will always have a valid tree), it will take

\[
\displaystyle\sum_{A \in P(T')}\left( c_A + |A|-| \displaystyle\bigcup_{B \in \delta^{-1}(A)} B| - 1 \right)
\]
\[
= -|P(T')| + \displaystyle\sum_{A \in P(T')}\left( c_A + |A|-| \displaystyle\bigcup_{B \in \delta^{-1}(A)} B| \right)
\]
\[
= |P(T)|-|P(T')| + \displaystyle\sum_{A \in P(T')}\left(|A|-| \displaystyle\bigcup_{B \in \delta^{-1}(A)} B| \right)
\]

\[
= \left( \displaystyle\sum_{A \in P(T')} |A| \right) - \left( \displaystyle\sum_{A \in P(T)} |A| \right) + |P(T)|-|P(T')|
\]

\[
= f(T') - f(T).
\]

Therefore, $d_\HP(T,T') \le f(T') - f(T)$.

We now observe that, by Theorem \ref{OneMove}, each binding can only increase or decrease the rank by $1$. Hence there is a lower bound on $d_\HP(T,T')$ of the difference between their ranks, so $d_\HP(T,T') = f(T') - f(T)$.
\end{proof}

\begin{corollary}
\label{DistanceEstimate}
If $T, T' \in RP(X)$, then 
\[|f(T)- f(T')| \le d_\HP(T,T') \le f(T)+f(T').\]
\end{corollary}

\begin{proof}
 That $|f(T)- f(T')| \le d_\HP(T,T')$ follows immediately from Theorem \ref{t:verticaldistance}. To see that $d_\HP(T,T') \le f(T)+f(T')$, observe that one can always get from $T$ to $T'$ by taking a path of down-moves to the star tree, then a path of up-moves to $T'$. Hence for any $T,T' \in RP(X)$, by Theorem \ref{t:verticaldistance} we have $d_\HP(T,T') \le f(T)+f(T')-2f(S) = f(T)+f(T')$. 
\end{proof}

We now derive some results on the diameter and neighbourhood of $RP(X)$ under $d_\HP$.

\begin{theorem}
\label{t:fbounds}
If $|X|=n$ and $T \in RP(X)$, then $0 \le f(T) \le \frac{(n-1)(n-2)}{2}$, with bounds tight and every integer value achieved by some tree in $RP(X)$. Equivalently, if $|X|=n$, $\mathcal{H}(X)$ is a graded poset with rank function $f$ and maximum rank $ \frac{(n-1)(n-2)}{2}$.
\end{theorem}

\begin{proof}
By Theorem \ref{t:verticaldistance} if $T <_\HP T'$, then $f(T) < f(T')$. Also, by Theorem \ref{OneMove} the function $f$ is compatible with the covering relation, so $\mathcal{H}(X)$ is a graded poset with rank function $f$.

Minimal $f$ is achieved by the star tree $S$ (as down-moves decrease $f$), which has $f(S)=0$.

Elements with maximal $f$ must be binary trees, because they are maximal in the poset and up-moves increase $f$. For all binary trees, $|H(T)|=2n-1$. We claim that caterpillar trees have maximal $f$, and we know for any caterpillar tree $C$, $f(C) = \frac{(n-1)(n-2)}{2}$. To see that caterpillar trees have maximal $f$, suppose you have some cluster $C$ of size $k$ that does not have a subcluster of size $k-1$. Observe that the `contribution' to $f$ of a cluster is strictly bounded above by the contribution of the cluster that contains it. There has to be at most two inclusion-maximal subclusters or we could make a binding, so call them $B_1,B_2$. Then the sum of the sizes of subclusters of $B_1$ has to be be at most $|B_1|-1+|B_2|-1 \le k-3$. But we could replace $B_1$ and $B_2$ by $B_1 \cup B_2$ plus one other element, without changing any of the structure below, and that has size $k-2$. The claim follows.

Hence the maximum value of $f(T) = \frac{n^2+3n-2}{2} - (2n-1) = \frac{(n-1)(n-2)}{2}$.

We can then observe that as we take any shortest undirected path from $S$ to a caterpillar tree, the value of $f(T)$ increases by $1$ each time.
\end{proof}

\begin{corollary}
If $|X|=n$ and the diameter of $RP(X)$ under $\le_\HP$ is $\Delta_\HP$, then 
\[
\frac{(n-1)(n-2)}{2} + 1 \le \Delta_\HP \le (n-1)(n-2).
\]
In particular, the diameter is $O(n^2)$.
\end{corollary}

\begin{proof}
As previously observed, one can always get from $T$ to $T'$ by a sequence of down-moves to the star tree, then up-moves to $T'$. Hence for any $T,T' \in RP(X)$, by Theorem \ref{t:verticaldistance} we have $d_\HP(T,T') \le f(T)+f(T')-2f(S)$. It follows by Theorem \ref{t:fbounds} that $\Delta_\HP \le (n-1)(n-2)$. 

We can also observe that for any caterpillar tree $C$ with inclusion-maximal proper cluster $X \backslash \{a\}$, any tree $T$ with a single proper cluster $ab$ for some leaf $b$ does not have a hierarchy-preserving map into $C$, and hence a shortest path from $C$ to $T$ must go from $C$ to the star tree to $T$, for a distance of $d(C,T) = f(T) - f(S) +1= \frac{(n-1)(n-2)}{2} +1$. Therefore $\frac{(n-1)(n-2)}{2} + 1 \le \Delta_\HP$ and the corollary holds.

Note that at least the upper bound can certainly be improved on, since no shortest path between a pair of binary trees with more than $3$ leaves includes the star tree.
\end{proof}

The size of the up-neighbourhood and down-neighbourhood of a given tree varies with the structure of the tree. We now investigate the maximum sizes of these neighbourhoods.

\begin{theorem}
Let $T \in RP(X)$, where $|X|=n$. Then the up-neighbourhood of $T$ contains at most $\frac{n(n-1)}{2}$ trees, with this value achieved only by the star tree.
\end{theorem}

\begin{proof}
We will show that deleting a proper cluster from $H(T)$ will increase the size of the up-neighbourhood of $T$. It follows that the tree with the largest up-neighbourhood is the star tree $S$, and we can observe that the up-neighbourhood of $S$ consists of the trees with a single proper cluster which is size $2$ - those obtained by binding any two leaves together. As there are $n$ leaves, there are $\binom{n}{2}=\frac{n(n-1)}{2}$ in the up-neighbourhood of $S$.

We can now show that deleting a proper cluster from $H(T)$ will increase the size of the up-neighbourhood of $T$. Suppose that we have some hierarchy $H(T)$, with some cluster $C$. Let $D$ be the cluster that $C$ is inclusion-maximal in (with the possibility $D=X$). Suppose $D$ has $k$ inclusion-maximal subclusters and that $C$ has $j$ inclusion-maximal subclusters. Then, suppose that $T$ has a total of $x$ possible bindings that do not include the inclusion-maximal clusters of $C$ or $D$. Suppose first that $k=2$. Then the inclusion-maximal subclusters of $D$ cannot bind (as they would form a cluster already in $H(T)$), for a total of $x + \binom{j}{2}$ trees in the up-neighbourhood of $T$, or just $x$ if $j=2$. But if we delete $C$ to form $T'$, we now have $x + \binom{j+1}{2}$ trees in the up-neighbourhood (that is, all of the previous bindings plus all of the bindings involving the inclusion-maximal subclusters of $C$), and is larger since $j >1$.

Now suppose $k > 2$. We can then immediately see that $T$ has a total of $x + \binom{j}{2} + \binom{k}{2}$ possible bindings, or just $x + \binom{k}{2}$ if $j=2$. However, once we have deleted $C$ to form $T'$, we have $x+\binom{j+k-1}{2}$ possible bindings, which is larger, as $j > 1$.

The result follows.
\end{proof}

\begin{theorem}
Let $T \in RP(X)$, where $|X|=n$. Then the down-neighbourhood of $T$ contains at most $2^{n-2} -1$ trees, with this value achieved only by trees with a single proper cluster, and that cluster is of the form $X \backslash \{a\}$, for some leaf $a$.
\end{theorem}

\begin{proof}
Suppose $T$ has some proper cluster $D$ with an inclusion-maximal proper subcluster $C$. Denote the inclusion-maximal subclusters of $C$ by $C_1,\dots,C_k$. Let $x$ be the number of valid unbindings of clusters that are not $C$ or $D$, $y$ be the number of valid unbindings of $D$, and $z$ the number of valid unbindings of $C$, so $T$ has a total of $x+y+z$ unbindings - that is, a down-neighbourhood of size $x+y+z$. Now, if we remove $C$ from $H(T)$, we claim that this increases the number of unbindings. This does not affect the number of valid unbindings of clusters that are not $C$ and $D$, so there are $x$ bindings of this type in $H(T) \backslash C$. Now, note that every valid unbinding of $D$ in $H(T)$ is a valid unbinding in $H(T) \backslash C$, as if $C$ is in a given partition, we can construct the same partition using the inclusion-maximal subclusters of $C$. Given that there is at least one partition here that we could not do before (deleting $D$ and replacing it by $C$ and $D \backslash C$), there are at least $y+1$ possible unbindings of $D$. We can also identify the $z$ unbindings of $C$ with $z$ unbindings of $D$ in the following way. Suppose $C$ is partitioned into $A$ and $B$ in $H(T)$. Then $D$ partitioned into $A$ and $B \cup (D \backslash C)$ is also a valid partition. It follows that there are at least $x+y+z+1$ trees in the down-neighbourhood of $H(T) \backslash C$, so the number of unbindings has been increased.

We can therefore consider only the hierarchies in which no proper cluster has a proper subcluster, that is, no proper subclusters intersect. Supposing there are $k$ proper clusters of size $i_1,\dots,i_k$ where $i_j \ge 2$ for all $j$ and $i_1+\dots+i_k \le n$, the number of splits of such a tree wil be

\[
\sum_{j=1}^k \stirling{i_j}{2}.
\]

Observe in particular that for trees with a single proper cluster, and that cluster is of the form $X \backslash \{a\}$, this becomes $\stirling{n-1}{2}$, and it follows from basic properties of the Stirling numbers of the second kind that 

\[
\sum_{j=1}^k \stirling{i_j}{2} \le \stirling{n-1}{2}.
\]

Hence trees of the form described have the largest possible number of splits, $\stirling{n-1}{2} = 2^{n-2} -1$, and the result follows.
\end{proof}

\begin{corollary}
The maximum neighbourhood size of a tree $T$ (the sum of the up- and down-neighbourhoods of $T$) is $O(2^{n-2})$.
\end{corollary}

\section{An upper bound on $d_{\HP}$}

In this section we present an algorithm for calculating an upper bound $e_\HP$ on the distance $d_\HP(T,T')$, because an exact calculation seems to be computationally expensive as the authors are yet to find an algorithm with subexponential run-time.  We will also show that the upper bound is quite often equal to the true distance (despite not being a metric itself --- see Observation~\ref{o:e.not.metric}). For instance, computational experiments show that $e_\HP=d_\HP$ in over $80\%$ of cases of pairs of trees on nine leaves (Section~\ref{s:comp.results}).

The method to find the upper bound depends on finding $\le_\HP$-maximal trees that have a hierarchy-preserving map into both $T$ and $T'$, and then finding a minimum path between $T$ and $T'$ that goes through one of these.  Of course, a geodesic path between $T$ and $T'$ need not visit any such $\le_\HP$-maximal tree, which is why this is only an upper bound. 

\begin{definition}
Let $T,T'$ be a pair of trees, and $\max_{\le_\HP}(T,T')$ be the set of trees $T_i$ in $RP(X)$ that are $\le_\HP$-maximal subject to the condition that $T_i \le_\HP T$ and  $T_i \le_\HP T'$. Then $e_\HP(T,T')$ is defined to be $\min(f(T) + f(T') - 2f(T_i))$ across all trees in $\max_{\le_\HP}(T,T')$.
\end{definition}

To find these, we will look at hierarchy-preserving maps in a different way, involving the following new definitions. 

\begin{definition}\label{d:multi-hier}
A \textit{multi-hierarchy} $\mathcal{M}$ on a set $X$ is a 
set of tuples $(A,i)$ (referred to as \textit{multi-clusters})
where $A\subseteq X$,
 and $i$ is a positive integer, 
with the following properties:

\begin{enumerate}
\item $\mathcal{M}$ contains both the tuple $(X,1)$ and all singleton tuples $(\{x\},1)$ for $x \in X$.
\item Let $(H_1,i),(H_2,j)$ be a pair of tuples in $\mathcal{M}$. Then $H_1 \cap H_2 = \varnothing$, $H_1 \subseteq H_2$ or $H_2 \subseteq H_1$.
\item The set of elements in $\mathcal{M}$ that share the same first entry $A$, say, $(A,i_1),\dots,(A,i_k)$ are numbered sequentially from $1$ to $k$ in the second entry.
\end{enumerate} 
The set of multihierarchies on a set $X$ will be denoted $\MRP(X)$. 
If $(A,i),(B,j)\in\MRP(X)$, we write $(A,i) \subseteq_M (B,j)$
if either $A \subset B$, or $A=B$ and $i \le j$.  In the latter case, if $i=j$, we write $(A,i) =_M (B,j)$.  
Define $(A,i) \subset_M (B,j)$ similarly except $i \ne j$.

Finally, if there is a multi-cluster $(A,i) \in \mathcal{M}$ where $A$ is a proper cluster on $X$, call $(A,i)$ a \textit{proper multi-cluster}.
\end{definition}

Note in particular that for any multi-hierarchy on $X$, there is a hierarchy on $X$ obtained by taking the \textit{support} of $\mathcal{M}$, denoted $supp(\mathcal{M})$ and defined by 
\[
supp(\mathcal{M})=\{A \mid (A,1) \in \mathcal{M}\}.
\] 
This is of course not a one-to-one correspondence as there can be many multi-hierarchies with the same support.

\begin{definition}\label{d:mhpm}
Let $T \in RP(X)$ and $\mathcal{M} \in \MRP(X)$. Then $\delta : H(T) \ra \mathcal{M}$ is a \textit{multi-hierarchy-preserving} map if the following properties hold for all $A,B \in H(T)$:
\begin{enumerate}
\item \textbf{Enveloping:} If $\delta(A) = (A',i)$, then $A \subseteq A'$, and
\item \textbf{Subset-Preserving:} $A \subset B$ implies that $\delta(A) \subset_M \delta(B)$.
\end{enumerate}
The set of trees with a {multi-hierarchy-preserving} map into $\mathcal{M}$ is denoted $\MHP(\mathcal{M})$.
\end{definition}

The reason for introducing these definitions is that for an algorithm to compute potential $\le_\HP$-maximal elements of $\HP(P)$, we require a systematic way of finding them. We will do this by taking certain intersections (see the algorithm below) of the clusters of $H(T)$ and $H(T')$, which unfortunately will not necessarily be a hierarchy. Observe that many of our results for hierarchy-preserving maps have an equivalent result for multi-hierarchy-preserving maps, proven in much the same way.

\begin{lemma}[Multi-hierarchy equivalent of Lemma \ref{UnionHP}]
\label{BindingMHP}
Let $T \in RP(X), \mathcal{M} \in \MRP(X)$, with a multi-hierarchy-preserving map $\delta: H(T) \ra \mathcal{M}$. Suppose $A$ and $B$ are distinct inclusion-maximal subclusters of some third distinct cluster $D$ in $H(T)$, where $D \ne A \cup B$ and $\delta(A) \subseteq_M \delta(B)$. Then the binding $T_{A\cup B}^D \in \MHP(\mathcal{M})$.
\end{lemma}

\subsection{Forming a multi-hierarchy from two trees}

Algorithm~\ref{a:makemulti} takes the hierarchies of two trees to produce a multi-hierarchy.

\begin{algorithm}[ht]
\caption{MAKEMULTI: producing a multi-hierarchy from two trees}\label{a:makemulti}
\begin{algorithmic}[1]
\Require $T,T'$ trees.
\State $\mathcal M\gets\emptyset$.
\While{$H(T)$ and $H(T')$ are non-empty}
\ForAll{maximal clusters $A_i \in H(T)$ and $B_j \in H(T')$} 
	\If{$C=A_i \cap B_j$ is non-empty} 
		\State $\mathcal{M}\gets\mathcal M\cup\{(C,k)\}$, where $k$ indicates the $k$-th occurrence of $C$
	\EndIf

	Delete all inclusion-maximal clusters of $H(T)$ and $H(T')$
\EndFor
\EndWhile
\end{algorithmic}
\end{algorithm}

We note here that as a tree has at most $2n$ clusters, the multi-hierarchy will contain at most $4n^2$ multi-clusters. In fact, this will generally not be a strict upper bound as we are only taking intersections of inclusion-maximal clusters with inclusion-maximal clusters, but it is sufficient for later showing that the algorithm has polynomial time complexity.

\begin{proposition}
The set $\mathcal{M}$ obtained from $T,T'$ using MAKEMULTI is a multi-hierarchy.
\end{proposition}

\begin{proof}
It is easily seen that $\mathcal{M}$ contains $(X,1)$ and all singleton tuples. The second entry of repeated elements being sequential from $1$ to $k$ is also obvious. Hence we just have to check requirement $(2)$ of Definition~\ref{d:multi-hier}.

Let $(A,i)$ and $(B,j)$ be two multi-clusters of $\mathcal{M}$ produced by the algorithm, and suppose that $A \cap B$ is non-empty. Suppose $(A,i)$ was obtained by taking the intersection of $A_1$ and $B_1$, and that $B$ was obtained by taking the intersection of $A_2$ and $B_2$. Now, since $A \cap B$ is non-empty, it follows that $A_1$ and $A_2$ have a non-empty intersection, and similarly for $B_1$ and $B_2$. It follows that either $A_1 \subseteq A_2$ or $A_1 \supset A_2$. Without loss of generality, suppose $A_1 \subseteq A_2$. Then $A$ was obtained on either the same step as $B$, or a subsequent step. If produced on the same step, it follows that $A_1=A_2$ and $B_1=B_2$, as inclusion-maximal elements have non-empty intersection with each other. Therefore $A=B$. Otherwise, if $A$ was obtained on a subsequent step, then $A_1 \subseteq A_2$ and $B_1 \subseteq B_2$ and so $A_1 \cap B_1 \subseteq A_2 \cap B_2$, and thus $A \subseteq B$. It follows that the set of clusters in $\mathcal{M}$ is a multi-hierarchy. 
\end{proof}

As the resulting set of tuples from the algorithm is a multi-hierarchy, determination of a $\le_\HP$-maximal element of $\HP(T,T')$ can be equivalently recognised as determination of a $\le_\HP$-maximal tree in $\MHP(\mathcal{M})$, where $\mathcal{M}$ is the multi-hierarchy obtained from $T,T'$.

\begin{example}
Consider the trees $T$ and $T'$  on the set $X= \{a,b,c,d,e,f,g\}$ so that $P(T) = ab, abcde,abcdef$ and $P(T') = ab, abcde, abcdeg$. Then the proper multi-clusters of the multi-hierarchy obtained from $T,T'$ are $\{ (abcde,1),(abcde,2),(ab,1) \}$ and the proper clusters in $supp(\mathcal{M})$ are $\{abcde,ab\}$.
\end{example}

\begin{example}
Suppose $\mathcal{M}$ is obtained via the algorithm from $T,T'$ and has a support corresponding to the hierarchy $H(T)$ of the tree $T$. Then if $\mathcal{M}=\{(A,1)| A \in H(T) \}$, the unique $\le_\HP$-maximal tree in $\MHP(\mathcal{M})$ is $T$ itself, and so $e_\HP(T,T') = d_\HP(T,T)+d_\HP(T',T)= f(T)+f(T') - 2f(T)$.
\end{example}

\begin{lemma}
\label{t:KnockOne}
Let $\mathcal{M}$ be the multi-hierarchy consisting only of $\{(A,1),\dots,(A,k)\}$ for $A \ne X$ . Then, the maximum value of $f(T)$ for $T \in \MHP(\mathcal{M})$ is 

\begin{equation*}
f(T)=\begin{cases}
k|A|- \frac{k(k+3)}{2}, & \text{if $|A| > k$}\\
\frac{(|A|-1)(|A|-2)}{2}, & \text{if $|A| \le k$.}\\
\end{cases}
\end{equation*}

\end{lemma}

\begin{proof}
Let $T \in \MHP(\mathcal{M})$. First suppose there is some cluster $C$ with more than two inclusion-maximal subclusters. Let two of them be $A,B$, and we can immediately see by Theorem~\ref{t:Union2OfMore} that $T_{A \cup B}^D \in \MHP(\mathcal{M})$ and $T \le_\HP T_{A \cup B}^D$, so $T$ is not $\le_\HP$-maximal. We can therefore assume every cluster of $T$ has at most $2$ inclusion-maximal subclusters.

Now, suppose that $C$ is an inclusion-minimal cluster of $T$ with respect to the requirement that $C$ has two inclusion-maximal clusters, neither of which is a singleton. Let the two inclusion-maximal clusters be $A$ and $B$. It follows that $f(T|_C) = \frac{(|A|-1)(|A|-2)}{2} + \frac{(|B|-1)(|B|-2)}{2}$, which is maximised if $|A|=1$ or $|B| = 1$. Therefore $T$ can only have maximal $f(T)$ if there is no non-singleton cluster that does not have a singleton subcluster. 

Therefore, the maximal possible value of $f(T)$ is achieved by mapping $A$ into $(A,1)$, then removing one element from $A$ for each mapping into $(A,2),(A,3)$, etc. The result follows.
\end{proof}

\begin{example}
If $\mathcal{M}$ is the multi-hierarchy obtained from $T,T'$, then, perhaps counterintuitively, it is not true in general that there exists a $\le_\HP$-maximal element of $\HP(T,T')$ that is a refinement of $supp(\mathcal{M})$ that has maximal $f$. Consider $\mathcal{M} = \{(abcdef,1),(abcdef,2), (abcdef,3), (ab,1),(cd,1),(ef,1) \}$. Then the maximum value of $f(T)$ is $23$ with e.g. $\{abcdef,abcde,abcd,ab,cd\}$, but the maximum value achievable with $T$ a refinement of $supp(\mathcal{M})$ is $f(T)=20$ with e.g. $\{abcdef,abcd,ab,cd,ef \}$.
\end{example}

We use Lemma \ref{t:KnockOne} as inspiration for the next algorithm, in Section~\ref{s:maxtree.alg}. In particular, that whenever MAKEMULTI produces a repeated cluster (i.e. a multi-cluster $(A,i)$ with $i>1$), we must delete one leaf from our cluster.

\subsection{Finding a $\le_\HP$-maximal tree in $\HP(T,T')$ using the multi-hierarchy of $T,T'$.}\label{s:maxtree.alg}

\begin{algorithm}[ht]
\caption{MAXTREE: an algorithm to find a maximal tree in $\HP(T,T')$ with maximal rank.}
\label{a:maxtree}
\begin{algorithmic}[1]
\Require The multi-hierarchy $\mathcal M$ obtained from $T$ and $T'$.
\State $T''\gets$ star tree.
\ForAll{$(A,i)\in\mathcal M$}
	
	Let $A'$ be the unique largest subcluster of $A$ for which $H(T'') \cup \{A'\}$ is a hierarchy.
	\If{$A'\not\in H(T'')$}
	\State $H(T'')\gets H(T'')\cup\{A'\}$.
	\ElsIf{$A'\in H(T'')$} 
		\If{$|A'|>1$}
		choose $x\in A'$
		\State $H(T'')\gets H(T'')\cup \{A'\setminus\{x\}\}$.
		\EndIf
	\EndIf
\EndFor

By iterating over all possible choices in line 6, we will find all $\le_\HP$-maximal trees in $\HP(T,T')$ (or equivalently $\MHP(\mathcal{M})$), and we take the tree with the highest rank.

\end{algorithmic}
\end{algorithm}

\begin{proposition}
The algorithmic complexity of determining the upper bound $e_\HP$ to $d_\HP(T,T')$ is polynomial.
\end{proposition}

\begin{proof}
Calculation of the rank $f(T)$ of $T$ is linear because there are at most $n$ clusters in a tree.

Calculation of the multi-hierarchy via MAKEMULTI (Algorithm~\ref{a:makemulti}) involves a linear number of intersections, and intersections can be done in linear time. Hence calculation of the multi-hierarchy is quadratic.

The only part of MAXTREE (Algorithm~\ref{a:maxtree}) that allows for choice is determining which elements to remove when we have repeated clusters. There are at most $4n^2$ multi-clusters in a multi-hierarchy obtained from two trees, and each cluster has a maximum of $n$ elements that we can choose to remove. Hence there is a maximum of $4n^3$ possible choices for a given multi-hierarchy, so iterating over all possible choices and checking $f(T)$ for each one will be polynomial in time complexity.
\end{proof}

\begin{example}
Unfortunately, $e_\HP$ is not equal to $d_\HP$ in general, as the following example demonstrates. Let $T$ and $T'$ be trees on $X=\{a,b,c,d,e,f,g\}$ with $P(T) = \{abc,de\}$ and $P(T') = \{ae,bdf \}$. Then the star tree is the unique tree with a hierarchy preserving map into both $T$ and $T'$, so the algorithm gives a distance of $e_\HP(T,T') = d_\HP(T,S) + d_\HP(T',S) = 3+3=6$. However, it is not difficult to find a path of length $4$ from $T$ to $T'$ in $\HH(X)$. For example, let $U_1,U_2,U_3$ be trees with $P(U_1) = \{ab,de\}, P(U_2) = \{abde\}$ and $P(U_3) = \{ae,bd\}$. Then the path $T,U_1,U_2,U_3,T'$ is one such path.
\end{example}

\begin{observation}\label{o:e.not.metric}
The above example also shows that $e_\HP$ is not a metric, because it fails the triangle inequality: we have $e_\HP(T,U_2)=e_\HP(U_2,T')=2$, but $e_\HP(T,T')=6$.
\end{observation}

\section{Computational results}\label{s:comp.results}

We have implemented the algorithms required to compute $e_\HP$, and in this section present some preliminary results.  Because MCMC algorithms often examine only binary trees, we explore both all of $RP(X)$ and also $BRP(X)$, the set of binary trees.  

A na\"ive algorithm to calculate the true distance $d_\HP$ (by checking all trees along all possible paths shorter than $e_\HP$, with some optimisations) can be used for trees on up to nine leaves, although the same approach for ten or more leaves can be very slow (for some pairs of trees over thirty minutes).
The algorithm, implemented in Python, can be found at \cite{Hendriksen2019}.

\subsection{Comparison of the upper bound $e_\HP$ with the true distance $d_\HP$. }

Figure \ref{9leaftest} shows the results of an experiment on $100$ random pairs of trees with $9$ leaves. The data indicate that the upper bound is reasonably accurate, with $e_\HP$ and $d_\HP$ being equal in $77\%$ of cases. The mean upper bound distance in this simulation was $9.87$, while the mean true distance was $9.39$. The biggest difference between the upper bound and the true distance was $4$.

\begin{figure}[ht]
\centering
\includegraphics{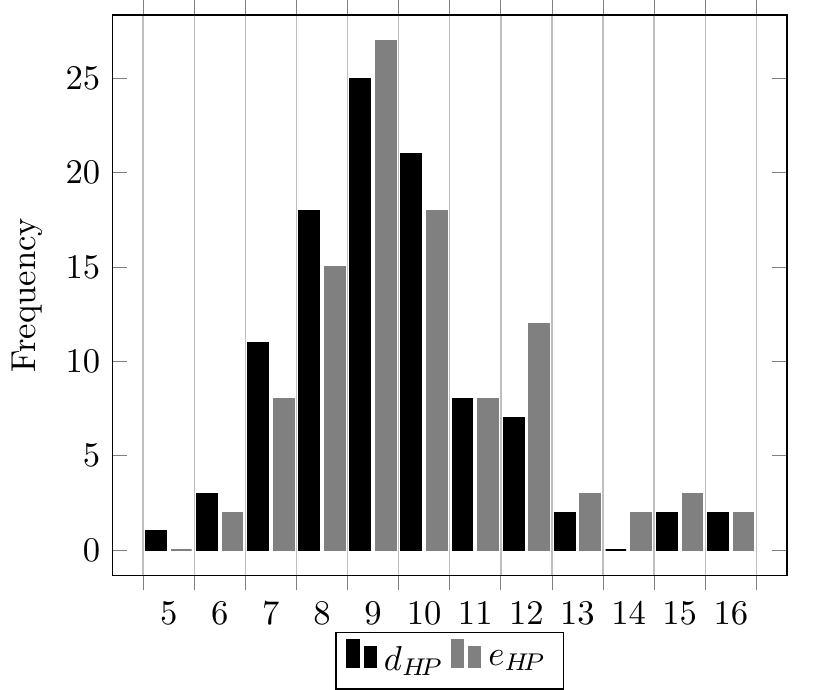}
\caption{A comparison of $e_\HP$ with $d_\HP$, on trees with $n=9$ leaves.}
\label{9leaftest}
\end{figure}

On the same data set, we also investigated how the proportion of $e_\HP$ values of a given distance were related to the value of $e_\HP$, with results given in Figure \ref{9leafperc}. Overall it appears that the larger the $e_\HP$, the less accurate the distances are, with the abrupt increase at distances $15$ and $16$ likely due to small sample sizes at this distance. We were unable to confirm this due to the exponential time that it takes our current algorithm to find $d_\HP$.

\begin{figure}[ht]
\centering
\includegraphics{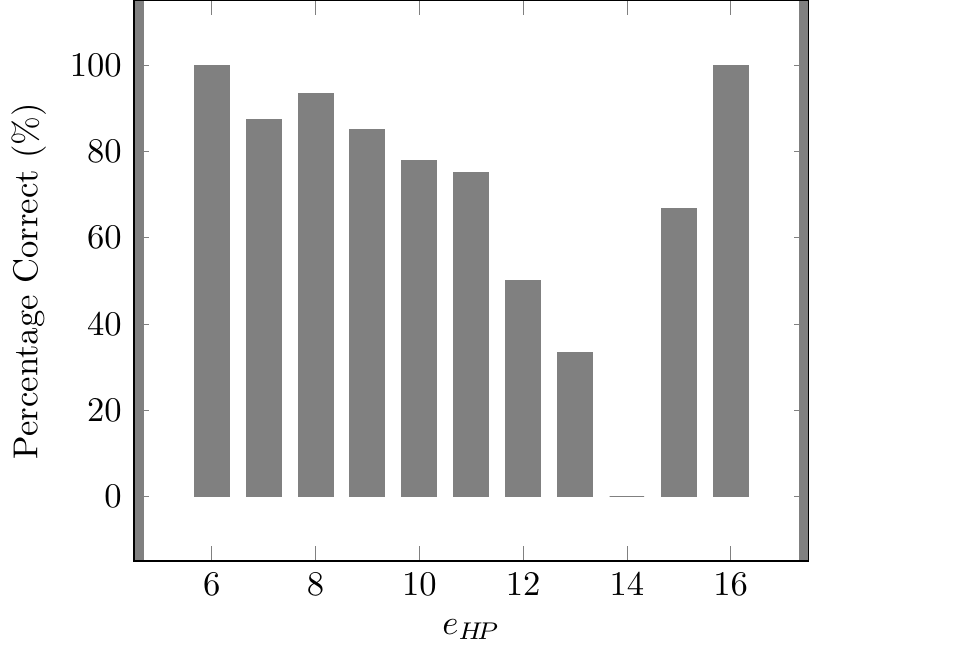}
\caption{A comparison of $e_\HP$ with the proportion of values of $e_\HP$  for which $e_\HP = d_\HP$, on trees with $n=9$ leaves.}
\label{9leafperc}
\end{figure}

\subsection{Experimental results on the upper bound $e_\HP$.}

Table~\ref{f:SimStats} shows some representative distance statistics for the upper bound $e_\HP$ on the distance.

\begin{table}
\begin{center}
\begin{tabular}{c| c c|| c c|} 
\cline{2-5}
& \multicolumn{2}{|c||}{RP(X)} & \multicolumn{2}{|c|}{BRP(X)} \\
\hline
\multicolumn{1}{ |c | }{n} & Average $e_\HP$ & Maximum $e_\HP$ & Average $e_\HP$ & Maximum $e_\HP$ \\ [0.5ex] 
\hline\hline
\multicolumn{1}{ |c | }{4} & 2.587 & 4 & 3.0 & 4 \\ 
\hline
\multicolumn{1}{ |c | }{5} & 4.645 & 8 & 5.525 & 8 \\ 
\hline
\multicolumn{1}{ |c | }{6} & 5.294 & 12 & 8.440 & 12 \\ 
\hline
\multicolumn{1}{ |c | }{7} & 6.990 & 16 & 10.123 & 16\\ 
\hline
\multicolumn{1}{ |c | }{8} & 8.752 & 17 & 12.900 & 19\\ 
\hline
\multicolumn{1}{ |c | }{9} & 10.708 & 21 & 15.883 & 24\\ 
\hline
\multicolumn{1}{ |c | }{10} & 12.695 & 24 & 18.983 & 29\\ 
\hline
\multicolumn{1}{ |c | }{20} & 35.719 & 57 & 56.344 & 74\\ 
\hline
\multicolumn{1}{ |c | }{40} & 91.662 & 123 & 151.527 & 176 \\ 
\hline
\end{tabular}
\end{center}
\caption{
Distance statistics for pairs of trees with each number of leaves. 
For $|X|\le 6$ (resp. $|X|\le 5$) these statistics represent calculations over \emph{all} pairs of trees in $RP(X)$ (resp. $BRP(X)$).  For larger leaf sets the results are the outcome of testing a sample of $20,000$ random pairs of trees.}
\label{f:SimStats}
\end{table}

The Average Distance column indicates the average $e_\HP$ between pairs, to three decimal places. These are provided as a baseline from which to judge the distance for a given pair of trees. 

The Maximum Distance column shows the maximum recorded $e_\HP$ between a pair of trees. Note that all trees that are the result of simulations only provide a lower bound on the maximum $e_\HP$, which is again an upper bound on the true $e_\HP$.

In particular, note that in Table~\ref{f:SimStats}, both the average and maximum $e_\HP$ on $BRP(X)$ are larger than those on all of $RP(X)$. Indeed, for $n=40$ on binary trees the average distance is larger than the maximum distance obtained for $n=40$ on all trees!  For such large trees the distributions of distances seem to radically diverge, as seen in Figure~\ref{f:DistHist}, which shows distances for 20,000 randomly selected pairs of trees.

\begin{figure}[ht]
\centering
\includegraphics[width = 0.95\textwidth]{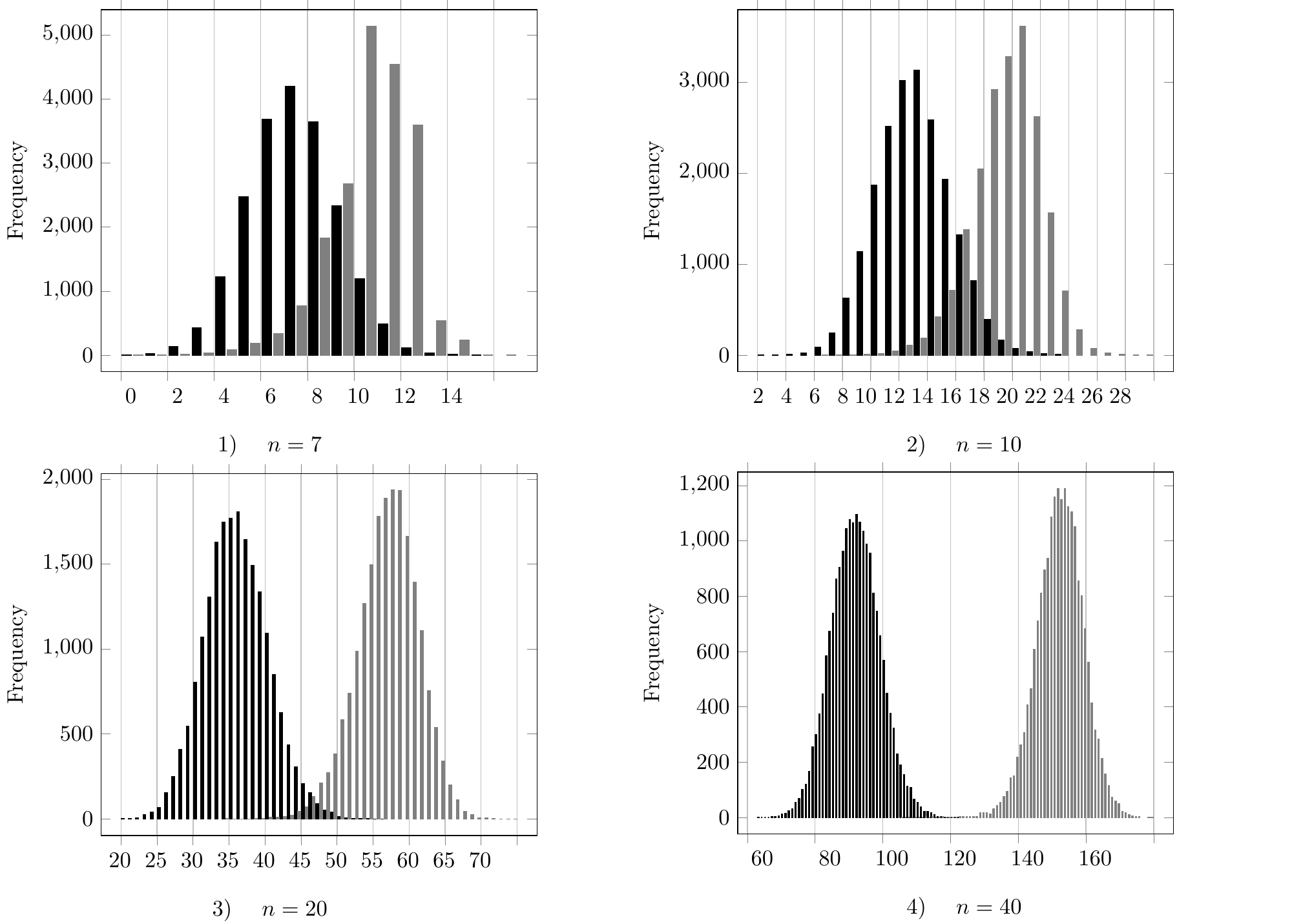}
\caption{Histograms of $e_\HP$ under 20,000 simulations of random pairs of trees with $n$ leaves. Simulations using trees randomly selected from all $\RP(X)$ in black, and $\BRP(X)$ in grey.}
\label{f:DistHist}
\end{figure}

Of course, the distributions don't \textit{actually} diverge, because after all the binary trees $\BRP(X)$ are a subset of the set of all trees $RP(X)$.  However the binary trees sit along the top of the very large Hasse diagram, since they are all of maximal rank (Prop~\ref{p:max.elts.RPX}), so the range of potential distances between them is therefore higher than any pair of nonbinary trees (Corollary \ref{DistanceEstimate}).  It is therefore, heuristically at least, unsurprising that the distances are correspondingly higher.

Part of the explanation for the apparent divergence of the distributions seen in Figure~\ref{f:DistHist} in the 40 leaf case is that the binary trees are such a small proportion of the total number of trees that when selecting a pair of random trees one almost never selects a pair of binary trees.  

\DeclareRobustCommand{\stirling}{\genfrac\{\}{0pt}{}}

In the sampling, trees are selected by randomly partitioning the set of leaves, and successively partitioning the components of the partition until all components have cardinality 1 (the leaves).  To select a binary tree, each successive partition must be a partition into exactly two components.  The probability of doing this is the number of partitions of 40 into two parts divided by the total number of partitions into any number of parts $k$.  These are counted by the Stirling numbers of the second kind, $\stirling{n}{k}$.  So the probability of even the first partition (immediately below the root) being binary is just
\[\frac{\stirling{40}{2}}{\sum_{k=2}^{40}\stirling{40}{k}},
\]
which is approximately $3.49\times 10^{-24}$.   To select a fully binary tree one would need to continue to choose further partitions into two parts at each point.  

It is perhaps worth noting the symmetry of the distributions shown in Figure~\ref{f:DistHist}, which suggest that the metric $e_\HP$ avoids the skew that affects the Robinson-Foulds metric.

\section{Discussion}

The new metric on phylogenetic tree space introduced in this paper has several interesting properties that may make it valuable for biological applications. 

First of all, it is a cluster-similarity metric, so the notion of distance between two trees corresponds to the similarity of their hierarchies.  This in itself is a valuable property in terms of comparisons of trees that have arisen under related processes, such as gene trees in the presence of incomplete lineage sorting.

Second, in contrast to other cluster-similarity metrics, this metric has a simple local operation to move around tree space, ensuring easy calculation of neighbourhoods.  This feature, coupled with the cluster-similarity property, can be expected to help with MCMC searches of tree-space around trees of similar hierarchies. 

And third, the distribution of distances on a given tree space appears to be quite symmetric, and also to have a reasonable spread of values.  This will be valuable in choosing trees from a set that are closest to each other or to a special tree (such as a purported species tree), in a way consistent with their hierarchies, and also makes it capable of distinguishing trees in a way required in the biological studies mentioned in the Introduction.

A primary goal for future study would be to either find an efficient method for calculating $d_\HP$ exactly, or a proof that any algorithm to calculate $d_\HP$ must have expoenetial runtime. If the complexity of this calculation is found to be high, results regarding the accuracy of the upper bound $e_\HP$ would prove useful. It would also be interesting to find tighter bounds for many of the results in this paper. For instance, under $d_\HP$, the diameter of $RP(X)$ and the neighbourhood size of a given tree $T$ can almost certainly be given better bounds.

It may be that the ranks of trees are able to provide additional information for estimating tree distances. For instance, Corollary \ref{DistanceEstimate} allows one to estimate distances between trees quite well if one or both trees have small rank. Further, it is not difficult to show that for any pair of binary trees $T,T'$, the distance $d_\HP(T,T') < f(T)+f(T')$ - note the strict inequality. Hence further research into the relationship between the ranks of trees and the distances between them may be fruitful.

Finally, the notion of hierarchy-preserving maps may be of independent mathematical interest. It is one of many possible generalisations of refinement, and as such is compatible with the notion. To our knowledge, the induced partial order and the concept of binding are both also new and may provoke further interest in the mathematical community.

\section*{Acknowledgements}
The authors thank the reviewers of the manuscript for many valuable suggestions, and Alexei Drummond and Bill Martin for helpful discussions while this manuscript was in preparation, at the New Zealand Phylogenetics meeting in February 2019.  The first author thanks Western Sydney University for its support in the form of an Australian Postgraduate Award during this research.

\bibliographystyle{plainnat}

\providecommand{\bysame}{\leavevmode\hbox to3em{\hrulefill}\thinspace}
\providecommand{\MR}{\relax\ifhmode\unskip\space\fi MR }

\end{document}